\documentclass[submission,copyright,creativecommons]{eptcs}
\providecommand{\event}{AFL 2014} 
\usepackage{breakurl}             

\usepackage{amsmath,amssymb,amsthm}
\usepackage{tikz}
\mathchardef\mhyphen="2D

\title{Synchronizing weighted automata}
\author{Szabolcs Iv\'an
\institute{University of Szeged, Hungary}
\email{szabivan@inf.u-szeged.hu}
}
\def\titlerunning{Synchronizing weighted automata}
\def\authorrunning{Szabolcs Iv\'an}

\def\bK{{\mathbf{K}}}
\def\bN{{\mathbf{N}}}
\def\bZ{{\mathbf{Z}}}
\def\bB{{\mathbf{B}}}

\theoremstyle{plain} 
\newtheorem{theorem}{Theorem} 
 
\newtheorem{proposition}{Proposition}

\theoremstyle{definition} 
\newtheorem{definition}{Definition}{}

\theoremstyle{remark} 
\newtheorem{remark}{Remark} 

\begin{document}

\maketitle

\begin{abstract}
We introduce two generalizations of synchronizability to automata with transitions weighted in an arbitrary semiring $\bK=(K,+,\cdot,0,1)$.
(or equivalently, to finite sets of matrices in $\bK^{n\times n}$.)
Let us call a matrix $A$ location-synchronizing if there exists a column in $A$ consisting of nonzero entries such that all the other columns of $A$ are filled by zeros.
If additionally all the entries of this designated column are the same, we call $A$ synchronizing.
Note that these notions coincide for stochastic matrices and also in the Boolean semiring.
A set $\mathcal{M}$ of matrices in $K^{n\times n}$ is called (location-)synchronizing if $\mathcal{M}$ generates a matrix subsemigroup containing a (location-)synchronizing matrix.
The $\bK$-(location-)synchronizability problem is the following: given a finite set $\mathcal{M}$ of $n\times n$ matrices with entries in $\bK$,
is it (location-)synchronizing? 
Both problems are PSPACE-hard for any nontrivial semiring.
We give sufficient conditions for the semiring $\bK$ when the problems are PSPACE-complete and show several undecidability results as well, e.g. synchronizability is undecidable if $1$ has infinite order in $(K,+,0)$ or when the free semigroup on two generators can be embedded into $(K,\cdot,1)$.
\end{abstract}

\section{Introduction}

The synchronization (directing, reseting) problem of classical, deterministic automata is a well-studied topic with a vast literature (see e.g.~\cite{volkov} for a survey). An automaton $\mathcal{A}$ is \emph{synchronizable} if some word $u$ induces a constant function on its state set, in which case $u$ is a synchronizing word of $\mathcal{A}$.
Deciding whether an automaton is synchronizable can be done in polynomial time and it is also known that for synchronizable automata, a synchronizing word of length $\mathcal{O}(n^3)$ exists, where $n$ denotes the number of its states.
(The famous \v{C}ern\'y conjecture from the sixties states that this bound is $(n-1)^2$.)

The notion of synchronizability has been extended e.g. (in three different ways) to nondeterministic automata in~\cite{imreh}, to stochastic automata in~\cite{kfouri} and more recently in another way in~\cite{doyen}, to integer-weighted transitions in~\cite{larsen}.
To our knowledge, only ad-hoc notions have been defined so far, each for a particular underlying semiring.
We note that in~\cite{larsen} the notion has also been extended to timed automata as well.

In this paper we introduce several extensions of synchronizability to automata with transitions weighted in an arbitrary semiring $\bK=(K,+,\cdot,0,1)$. For states $p,q$ and word $u$, let $(pu)_q\in K$ denote the sum of the weights of all $u$-labeled paths from $p$ to $q$, with the weight of a path being the product of the weights of its edges, as usual.
Following the nomenclature of~\cite{larsen},
we call the automaton $\mathcal{A}$ \emph{location-synchronizable} if $\exists q,u$: $\forall p,r$\quad$(pu)_r\neq 0$ iff $r=q$ and \emph{synchronizable} if $\exists q,u,k\neq 0$: $\forall p,r$\quad $(pu)_q=k$ and $(pu)_r=0$ for each $r\neq q$.

As an equivalent formulation, let us call a matrix $A\in\bK^{n\times n}$ \emph{location synchronizing} if it contains a column entirely filled with nonzero values, and all its other entries are zero.
If in addition all the nonzero values are the same, we call $A$ \emph{synchronizing}.
Then, an instance of the synchronizability problems is a finite set $\mathcal{A}=\{A_i:1\leq i\leq k\}$ of matrices, each in $\bK^{n\times n}$. 
The family $\mathcal{A}$ is called (location) synchronizable if it generates a (location) synchronizing matrix.
The question is to decide whether the instance is (location) synchronizing.

Note that these notions coincide for stochastic automata and also in the Boolean semiring.
For unconstrained automata, both problems are $\mathbf{PSPACE}$-hard for any nontrivial semiring, and in any semiring,
the length of the shortest directing word can be exponential.
We give sufficient conditions for the semiring $\bK$ when the problems are in $\mathbf{PSPACE}$ (and hence are $\mathbf{PSPACE}$-complete)
and show several undecidability results as well.

\section{Notation}

A \emph{semiring} is an algebraic structure $\bK=(K,+,\cdot,0,1)$ where $(K,+,0)$ is a commutative monoid with identity $0$,
$(K,\cdot,1)$ is a monoid with identity $1$, $0$ is an annihilator for $\cdot$ and $\cdot$ distributes over $+$, i.e.
$0a=a0=0$, $(a+b)c=ac+bc$ and $a(b+c)=ab+ac$ for each $a,b,c\in K$. (When the context is clear, we usually omit the $\cdot$ sign.)
The case when $|K|=1$ is that of the trivial semiring; when $|K|>1$, the semiring is nontrivial.
Three semirings used in this paper are
the \emph{Boolean semiring} $\bB=(\{0,1\},\vee,\wedge,0,1)$
and the semirings $\bN$ and $\bZ$ of the natural numbers $\{0,1,2,\ldots\}$ and the integers $\{0,\pm 1,\pm 2,\ldots\}$
with the standard addition and product. Among these, only $\bZ$ is a ring since the other two have no additive inverses.
A semiring $\bK$ is zero-sum-free if $a+b=0$ implies $a=b=0$; is zero-divisor-free if $ab=0$ implies $a=0$ or $b=0$;
is positive if it is both zero-sum-free and zero-divisor-free;
is locally finite if for any finite $K_0\subseteq K$, the least subsemiring of $K$ containing $K_0$ (which is also called
the subsemiring of $K$ generated by $K_0$) is finite.

An \emph{alphabet} is a finite nonempty set, usually denoted $A$ in this paper.
When $n$ is an integer, $[n]$ stands for the set $\{1,\ldots,n\}$.
For a set $X$, $P(X)$ denotes its power set $\{Y:Y\subseteq X\}$.
For any alphabet $A$, the semiring of \emph{languages} over $A$ is $(P(A^*),\cup,\cdot,\emptyset,\{\varepsilon\})$ where
product is concatenation of languages, $KL=\{uv:u\in K,v\in L\}$ and $\varepsilon$ stands for the empty word.

When $\bK$ is a semiring and $n>0$ is an integer, then the set $\bK^{n\times n}$ of $n\times n$ matrices with entries in $\bK$
also forms a semiring with pointwise addition $(A+B)_{i,j}=A_{i,j}+B_{i,j}$ (for clarity, $A_{i,j}$ stands for the entry in the
$i$th row and $j$th column)
and the usual matrix product $(AB)_{i,j}=\sum_{k\in[n]}A_{i,k}B_{k,j}$.
The zero element is the null matrix $\mathcal{O}_{i,j}=0$
and the one element is the identity matrix $I_{i,j}=\begin{cases}1\textrm{, if }i=j\\0\textrm{, otherwise}\end{cases}$ in $\bK^{n\times n}$.

In this article we only take products of matrices, no sums and thus use the notion $\langle \mathcal{M}\rangle$ when
$\mathcal{M}\subseteq\bK^{n\times n}$ is a set of matrices for the least sub\emph{monoid} of the monoid $(\bK^{n\times n},\cdot,I_{n})$
containing $\mathcal{M}$. That is, $\langle\mathcal{M}\rangle$ contains all products of the form $M_1M_2\ldots M_k$ with
$k\geq 0$ and $M_i\in\mathcal{M}$ for each $i\in[k]$.

For a semiring $\bK$, alphabet $A$ and integer $n>0$, an \emph{$n$-state $\bK$-weighted $A$-automaton}
is a system $M=(\alpha,(M_a)_{a\in\Sigma},\beta)$
where $\alpha,\beta\in\bK^n$ are the \emph{initial} and \emph{final} vectors, respectively and for each $a\in A$,
$M_a\in \bK^{n\times n}$ is a \emph{transition matrix}. The mapping $a\mapsto M_a$ extends in a unique way to a homomorphism
$A^*\to \bK^{n\times n}$, $w\mapsto M_w$ with $M_{a_1\ldots a_k}=M_{a_1}\ldots M_{a_k}$. The automaton $M$ above associates to
each word $w$ a weight $M(w)=\alpha M_w\beta\in K$, where $\alpha$ is considered as a $1\times n$ row vector and $\beta$ as an
$n\times 1$ column vector. We usually do not specify the number $n$ of states explicitly and omit $\bK$ and $A$ when
the weight structure and/or the alphabet is clear from the context.

\section{Synchronizability in various semirings}

Classical nondeterministic automata (with multiple initial states but no $\varepsilon$-transitions)
can be seen as automata with weights in the Boolean semiring.
For any semiring $\bK$, a $\bK$-automaton $M=(\alpha,(M_a)_{a\in A},\beta)$ is
\begin{itemize}
\item \emph{partial} if there is at most one nonzero entry in each row of each transition matrix, and $\alpha$ has exactly one nonzero entry,
\item \emph{deterministic} if it is partial and there is exactly one nonzero entry in each row of each matrix $M_a$.
\end{itemize}
A classical deterministic automaton $M=(\alpha,(M_a)_{a\in A},\beta)$ is called
\emph{synchronizable} (directable, resetable etc) if there exists a word $w$ (called a synchronizing word of $M$) such that $M_w$
has exactly one column that is filled with $1$'s and all the other entries of $M_w$ are zero.
(Traditionally, this property is formalized as $w$ inducing a constant map on the state set.)

As an example, the $4$-state automaton $M=(\alpha,(M_a)_{a\in\{0,1\}},\beta)$ with arbitrary $\alpha$ and $\beta$ and with transition matrices
\[
M_0=\left(\begin{array}{llll}0&1&0&0\\0&0&1&0\\0&0&0&1\\1&0&0&0\end{array}\right),
\hfill M_1=\left(\begin{array}{llll}0&1&0&0\\0&1&0&0\\0&0&1&0\\0&0&0&1\end{array}\right)
\]
is synchronizable since for the word $100010001$, the transition matrix is
\[(M_1(M_0)^3)^2M_1=\left(\begin{array}{llll}0&1&0&0\\0&1&0&0\\0&1&0&0\\0&1&0&0\end{array}\right).\]

The notion celebrates its 50th anniversary this year -- a very popular and intensively studied conjecture in the area is that of
\v{C}ern\'y stating if an $n$-state classical deterministic automaton is synchronizable, then it admits a synchronizing word of length at most $(n-1)^2$.
We remark here that it is decidable in polynomial time (it's actually in $\mathbf{NL}$)
whether an input classical, deterministic automaton is synchronizable.

Synchronizability has been extended to nondeterminisic automata in~\cite{imreh} in three different ways.
Here we highlight the one entitled ``D3-directability'' there: a $\bB$-automaton $M=(\alpha,(M_a)_{a\in A},\beta)$
(that is, a classical nondeterministic automaton)
is called D3-directable if there exists a word $w$ such that $M_w$ has exactly one column that is filled with $1$'s and
all the other entries of $M_w$ are zero.
It is known (see e.g.~\cite{imreh}) that in general, the shortest synchronizing word of a synchronizable $n$-state $\bB$-automaton
can have length $\Omega(2^n)$ with $O(2^n)$ being an upper bound~\cite{gazdag}.
For partial $\bB$-automata, the best known bounds are $\Omega(\sqrt[3]{3}^n)$ and $O(n^2\sqrt[3]{4}^n)$, see~\cite{martyugin-bound,gazdag}.

In the next section of the paper we will frequently use the following results of~\cite{martyugin-pspace}:
\begin{theorem}
\label{thm-marty}
Deciding whether an input $\bB$-automaton is synchronizable is complete for $\mathbf{PSPACE}$.
The problem remains $\mathbf{PSPACE}$-complete when restricted to partial $\bB$-automata.
\end{theorem}

For the probabilistic semiring, in which case the weight structure is that of the nonnegative reals with the standard addition and product,
and the input automata's transition matrices are restricted to be stochastic, the notion has been also generalized by several authors:
\begin{itemize}
\item In~\cite{kfouri}, $M$ is synchronizable if there exists a word $w$ such that all the rows of $M_w$ are identical.
\item In~\cite{doyen}, $M$ is synchronizable if there exists a single infinite word $w$ such that for any $\epsilon>0$,
  there exists an integer $K_\epsilon$ such that for each finite prefix $u$ of $w$ having length at least $K_\epsilon$,
  in $M_u$ there is a column in which each entry is at least $1-\epsilon$.
\end{itemize}
The problem of checking synchronizability is undecidable in the former setting and ${\mathbf{PSPACE}}$-complete in the latter setting.

Most of these generalizations require (an arbitrary precise approximation of) a column consisting of ones and
zeros everywhere else in some matrix of the form $M_w$. In fact, under these conditions it is a simple consequence of the structure of
the semiring and the constraint on the automata that if in a row of a transition matrix $M_w$ there is exactly one nonzero element,
then it has to be $1$. (The Boolean semiring has only two elements, while in the probability semiring the stochasticity of the matrices
guarantee that the row sum is preserved and is one.)

The authors of~\cite{larsen} worked in the semiring $\bZ$, with a different semantics notion, though:
according to the notions of the present paper they worked in the semiring $P_f(\bZ)$, where the elements are finite sets of integers,
with union as addition and complex sum $X+Y=\{x+y:x\in X,y\in Y\}$ being product.
There two different notions of synchronizability are introduced: a matrix $M$ is \emph{location synchronizing} if there exists
a column in which each entry is nonzero, while all the other entries of the matrix are zeroes (recall that in this semiring
$\emptyset$ plays as zero)
and is \emph{synchronizing} if additionally the nonzero entries all coincide and map every possible starting vector $\alpha$
to some fixed vector (which is simply not possible in this semiring since this would require the presence of an $\mathcal{L}$-trivial element of the semiring).
An automaton $M$ is location synchronizable if there exists a word $w$ such that
$M_w$ is location synchronizing.
Regarding the complexity issues, location synchronizability is $\mathbf{PSPACE}$-complete (which is due to the fact
that $P_f(\bZ)$ is positive, cf. Proposition~\ref{prop-zsf-zdf})
and synchronizability is trivially false.

In this paper we extend the notion of synchronizability in spirit similar to~\cite{larsen}, covering most of the generalizations above
(the exception being the case of the probabilistic semiring, which seems to require a notion of metric).
\begin{definition}
Given a semiring $\bK$ and a matrix $M\in\bK^{n\times n}$, we say that $M$ is
\begin{itemize}
\item \emph{location synchronizing} if there exists a (unique) integer $i\in[n]$ such that $M_{j,k}\neq 0$ iff $k=i$;
\item \emph{synchronizing} if it is location synchronizing and additionally, $M_{j,i}=M_{1,i}$ for each $j\in[n]$ for the above index $i$.
\end{itemize}
A finite set $M_1,\ldots,M_k\in\bK^{n\times n}$ of matrices is \emph{(location) synchronizable} if they generate a (location) synchronizing matrix,
i.e. when $M_{i_1}M_{i_2}\ldots M_{i_t}$ is (location) synchronizing for some $i_1,\ldots,i_t\in[k]$, $t>0$.

A $\bK$-automaton is (location) synchronizable if so is its set of transition matrices.

We formulate the $\bK$-(location) synchronizing problem
($\bK\mhyphen\mathbf{Sync}$ and $\bK\mhyphen\mathbf{LocSync}$ for short)
as follows: given a finite set $\mathcal{M}=\{M_1,\ldots,M_k\}$ of matrices
in $\bK^{n\times n}$ for some $n>0$, decide whether $\mathcal{M}$ is (location) synchronizable?

(Clearly, this is equivalent to having a single $\bK$-automaton as input.)
\end{definition}

\section{Results on complexity of the two problems}

Given a semiring $\bK$, call a matrix $M\in\bK^{n\times n}$ a partial $0/1$-matrix if in each row there is at most one nonzero entry,
which can have only a value of $1$ if present, formally for each $i$ there exists at most one $j$ with $M_{i,j}\neq 0$ in which
case $M_{i,j}=1$ has to hold. Observe that the product of two partial $0/1$-matrices is still a partial $0/1$-matrix, being the
same in any semiring. Moreover, a partial $0/1$-matrix is synchronizing iff it is location synchronizing.
Thus the following are equivalent for any set $\mathcal{M}\subseteq\bK^{n\times n}$ of partial $0/1$-matrices:
\begin{enumerate}
\item $\mathcal{M}$ is synchronizable;
\item $\mathcal{M}$ is location synchronizable;
\item $\mathcal{M}$, viewed as a set of partial $0/1$-matrices over $\bB$, is synchronizable.
\end{enumerate}
Since by Theorem~\ref{thm-marty} the last condition is $\mathbf{PSPACE}$-hard to check, we immediately get the following:
\begin{proposition}
For any nontrivial semiring $\bK$, both $\bK\mhyphen\mathbf{Sync}$ and $\bK\mhyphen\mathbf{LocSync}$ are $\mathbf{PSPACE}$-hard.
\end{proposition}

\subsection{Decidable subcases}

First we make several (rather straightforward) observations on decidable subcases, generally involving finiteness conditions.

Of course if $\bK$ is finite, we get $\mathbf{PSPACE}$-completeness: 
\begin{proposition}
For any finite semiring $\bK$ both problems are in $\mathbf{PSPACE}$, thus are $\mathbf{PSPACE}$-complete.
\end{proposition}
\begin{proof}
Given an instance $\mathcal{M}=\{M_1,\ldots,M_k\}$ of the problem,
we store a current matrix $C\in\bK^{n\times n}$ initialized by the unit matrix $I_n$ of $\bK^{n\times n}$.
In an endless loop, we nondeterministically choose an index $i\in [k]$ and let $C:=CA_i$.
After each step we check whether $C$ is (location) synchronizing. If so, we report acceptance, otherwise
continue the iteration.

If $\bK$ is finite, storing an entry of $C$ takes constant space, so storing $C$ takes $O(n^2)$ memory,
as well as computation of the product matrix. In total, we have an $\mathbf{NPSPACE}$ algorithm which
is $\mathbf{PSPACE}$ by Savitch's theorem~\cite{papadimitriou}.
\end{proof}

\begin{proposition}
For any locally finite semiring $\bK$, both $\bK\mhyphen\mathbf{Sync}$ and
$\bK\mhyphen\mathbf{LocSync}$ are decidable, provided that addition and product of $\bK$ are computable.
\end{proposition}
\begin{proof}
Recall that a semiring $\bK$ is locally finite if any finite subset of $K$ generates a finite subsemiring of $\bK$.

Now given an instance $\mathcal{M}=\{M_1,\ldots,M_k\}$ of the problem,
let $X=\{{M_i}_{j,t}:i\in[k],j,t\in[n]\}\subseteq K$ stand for the finite set of the entries occurring in any of the matrices.
Then clearly, $\langle \mathcal{M}\rangle\subseteq \mathbf{X}^{n\times n}$ where $\mathbf{X}$ is the subsemiring of $\bK$
generated by $X$.
Since $\bK$ is finitely generated, this
implies $\langle\mathcal{M}\rangle$ is finite as well, hence there exists an integer $t$ such that $\langle\mathcal{M}\rangle
=\mathcal{M}^{\leq t}=\{M_{i_1}M_{i_2}\ldots M_{i_d}:d\leq t,i_1,\ldots,i_d\in[k]\}$ which can be chosen to be the least integer 
$t$ with $\mathcal{M}^{\leq t}=\mathcal{M}^{\leq t+1}$. Hence by computing the sets $\mathcal{M}^{\leq t}$ for $t=0,1,2,\ldots$
and reporting acceptance when a witness is found and rejecting the input when $\mathcal{M}^{\leq t}=\mathcal{M}^{\leq t+1}$
gets satisfied without finding a witness we decide the respective problem.

(Note that computability of addition and product is needed for the effective computation of the sets above.)
\end{proof}

\begin{proposition}
\label{prop-zsf-zdf}
For any positive semiring $\bK$, $\bK\mhyphen\mathbf{LocSync}$ is in $\mathbf{PSPACE}$.
\end{proposition}
\begin{proof}
For any positive semiring $\bK$ the mapping $\sigma:\bK\to\bB$ which maps $0$ to $0$ and all other elements of $K$ to $1$,
is a semiring morphism. Hence $\sigma$ can be extended pointwise to a semiring morphism $\sigma:\bK^{n\times n}\to\bB^{n\times n}$,
with $(\sigma(A))_{i,j}=\sigma(A_{i,j})$. Then, a matrix $A\in\bK^{n\times n}$ is \emph{location} synchronizing if and only if
$\sigma(A)$ is (location) synchronizing. Hence $\bK\mhyphen\mathbf{LocSync}$ can be reduced to $\bB\mhyphen\mathbf{Sync}$
via the polytime reduction $\{A_1,\ldots,A_k\}\mapsto\{\sigma(A_1),\ldots,\sigma(A_k)\}$, which is solvable in $\mathbf{PSPACE}$,
hence so is $\bK\mhyphen\mathbf{LocSync}$.
\end{proof}
\begin{remark}
One can use the above semiring morphism to decide any such property of matrices which cares only on the positions of zeroes
(i.e. when $M$ satisfies the property if and only if so does $\sigma(M)$).
Examples of such properties are \emph{mortality} (whether the all-zero matrix is generated),
and the \emph{zero-in-the-upper-left-corner} (whether a matrix with a zero in the upper-left corner is generated).
Thus both properties are in $\mathbf{PSPACE}$ for positive semirings
(and are in fact undecidable for the semiring $\bZ$, which is \emph{not} zero-sum-free).

Synchronizability, on the other hand, as well as the ``equal entries problem'' asking whether a matrix is generated having
the same entry at two specified positions, is not such a property. The latter is well-known to be undecidable in $\bN$
while the former is shown to be undecidable in Theorem~\ref{thm-n-sync}.
\end{remark}

\subsection{Undecidable subcases}
Now we turn our attention to undecidability results.

A well-known undecidable problem is the \emph{Fixed Post Correspondence Problem}, or FPCP for short: given a finite set
$\{(u_1,v_1),\ldots,(u_k,v_k)\}$ of pairs of nonempty words over a binary alphabet,
does there exist a nonempty index sequence $i_1,\ldots,i_t$, each $i_j$ in $[k]$, $t>0$ with $i_t=1$ (i.e. we fix the
\emph{last} used tile) such that
$u_{i_1}u_{i_2}\ldots u_{i_t}=v_{i_1}v_{i_2}\ldots v_{i_t}$? The problem is already undecidable for the fixed
constant $k=7$ (also, it's known to be decidable for $k=2$, see~\cite{halava} and has an unknown decidability status
for $3\leq k\leq 6$).

\begin{proposition}
For any semiring $\bK$ such that the semigroup $(\{a,b\}^*,\cdot)$ embeds into the multiplicative monoid $(K,\cdot,1)$ of $\bK$,
the $\bK\mhyphen\mathbf{Sync}$ problem is undecidable, even for two-state deterministic WFA with an alphabet size of
$8$ (i.e. for eight $2\times 2$ matrices when the question is viewed as a problem for matrices).
\end{proposition}
\begin{proof}
In order to ease notation, suppose $(\{a,b\}^*,\cdot)$ is a subsemigroup of $(K,\cdot,1)$.
For words $u,v\in\{a,b\}^+$, let us define the matrices 
$A(u,v)=\left(\begin{array}{ll}u&0\\0&v\end{array}\right)$ and $B(u,v)=\left(\begin{array}{ll}u&0\\v&0\end{array}\right)$.
Then a direct computation shows that
\begin{align*}
A(u_1,v_1)A(u_2,v_2)&=A(u_1u_2,v_1v_2),\\
B(u_1,v_1)A(u_2,v_2)&=B(u_1u_2,v_1v_2),\\
B(u_1,v_1)A(u_2,v_2)=B(u_1,v_1)B(u_2,v_2)&=B(u_1u_2,v_1u_2).
\end{align*}
Also, matrices $A(u,v)$ are not synchronizing while matrices $B(u,v)$ are synchronizing iff $u=v$.
Moreover, a product $B(u_1,v_1)X$ is synchronizing for $X\in\langle\cup_{u,v\in\{a,b\}^+}\{A(u,v),B(u,v)\}\rangle$ iff $u_1=v_1$.
Thus we can derive that a product of the form  $X_1(u_1,v_2)X_2(u_2,v_2)\ldots X_k(u_k,v_k)$ with each $X_i$ being either $A$ or $B$
and $u_i,v_i\in\{0,1\}^+$ is synchronizing iff there exists some $t\in[k]$ such that $X_t=B$, $X_{t'}=A$ for each $t'<t$ and
$u_1\ldots u_t=v_1\ldots v_t$ holds.

Hence, a reduction from FPCP to $\bK\mhyphen\mathbf{Sync}$ is given by the transformation
\[\{(u_i,v_i):i\in[k]\}\quad\mapsto\quad\{A(u_i,v_i):i\in[k]\}\cup\{B(u_1,v_1)\}.\]
Since FPCP is undecidable, so is $\bK\mhyphen\mathbf{Sync}$.
\end{proof}

Note that $(\Sigma^*,\cup,\cdot,\emptyset,\{\varepsilon\})$ is positive, so its location synchronization problem is decidable in polynomial space, while when $|\Sigma|>1$, its synchronization problem becomes undecidable.

Now we give a polynomial-time reduction from the $\bK$-mortality problem to both of the $\bK$-synchronization and the $\bK$-location synchronization problem. The $\bK$-mortality problem is actively studied for the case $\bK=\bZ$:
\begin{definition}
For a fixed semiring $\bK$, the $\bK$-mortality problem is the following: given a finite set $\mathcal{M}=\{M_1,\ldots,M_k\}$
of matrices in $\bK^{n\times n}$ for some $n>0$, does $\langle \mathcal{M}\rangle$ contain the null matrix $\mathcal{O}_n$?
\end{definition}

\begin{proposition}
For any semiring $K$, the $\bK$-mortality problem reduces to both of $\bK\mhyphen\mathbf{Sync}$
and $\bK\mhyphen\mathbf{LocSync}$.
Thus, in particular, when $\bK$-mortality problem is undecidable, so are both synchronizability problems.
\end{proposition}
\begin{proof}
Let $\mathcal{M}=\{M_1,\ldots,M_k\}$ be an instance of the $\bK$-mortality problem.
We define the matrices $A_i=\left(\begin{array}{ll}1&\mathbf{0}\\\mathbf{0}&M_i\end{array}\right)$, i.e. adding an all-zero top row and an all-zero
first row to each $M_i$, $i\in[k]$ and fill the upper-left corner by $1$. Also, we define $A_0=\left(\begin{array}{ll}1&\mathbf{0}\\\mathbf{1}&I_n\end{array}\right)$. We claim that the following are equivalent:
\begin{enumerate}
\item $\mathcal{O}_n\in\langle \mathcal{M}\rangle$;
\item $\mathcal{A}=\{A_i:0\leq i\leq k\}$ is synchronizable;
\item $\mathcal{A}$ is location synchronizable.
\end{enumerate}
Observe that each member of $\mathcal{A}$ is block-lower triangular with $1$ in the upper left corner, hence for any product
$A=A_{i_1}A_{i_2}\ldots A_{i_t}$ we have $A=\left(\begin{array}{ll}1&0\\X&M_{i_1}M_{i_2}\ldots M_{i_t}\end{array}\right)$ for some column vector $X$.
Note that in order to ease notation we define $M_0$ as the unit matrix $I_n$ and set $\mathcal{M}=\{M_0,\ldots,M_k\}$ -- since $I_n$ is not
synchronizing and is the unit element of $\bK^{n\times n}$, this neither affects mortality (of $\mathcal{M}$) nor synchronizability
(of $\mathcal{A}$).

Thus in particular the first column of any matrix $A\in\langle\mathcal{M}\rangle$ contains a nonzero entry, hence $A$ is (location) synchronizing
only if $M_{i_1}M_{i_2}\ldots M_{i_t}=\mathcal{O}_n$, in which case $\mathcal{M}$ is indeed a positive instance of the $\bK$-mortality problem,
showing iii)$\to$ i). For i)$\to$ii), let $A_{i_1}\ldots A_{i_t}=\mathcal{O}_n$, $t>0$, $i_j\in[k]$. Then $M:=M_{i_1}\ldots M_{i_t}=\left(\begin{array}{ll}1&\mathbf{0}\\\mathbf{0}&\mathcal{O}_n\end{array}\right)$, thus $A_0M=\left(\begin{array}{ll}1&\mathbf{0}\\\mathbf{1}&\mathcal{O}_n\end{array}\right)$ is a synchronizing matrix. Finally, ii)$\to$iii) is clear for any $\mathcal{A}$.
\end{proof}
In particular, since mortality is undecidable in $\bZ$, so are $\bZ\mhyphen\mathbf{Sync}$and $\bZ\mhyphen\mathbf{LocSync}$.

Our most involved result on undecidability is the following one:
\begin{theorem}
\label{thm-n-sync}
$\bN\mhyphen\mathbf{Sync}$ is undecidable.
Thus if $\bN$ embeds into $\bK$ (i.e. when $1$ has infinite order in $(K,+,0)$), then so is $\bK\mhyphen\mathbf{Sync}$.
\end{theorem}
\begin{proof}
We give a polynomial-time reduction from the FPCP problem to $\bN\mhyphen\mathbf{Sync}$.
This time we use the variant of FPCP in which the \emph{first} tile is fixed to $(u_1,v_1)$.
Let $\{(u_i,v_i):i\in[k]\}$ be an instance of the FPCP, $u_i,v_i\in\{0,1\}^+$.
For a nonempty word $u\in\{0,1\}^+$ let $\textrm{int}(u)$ be its value when considered as a \emph{ternary} number, i.e.
$\textrm{int}(a_{n-1}\ldots a_0)=\sum_{0\leq i<n}a_i3^i$.
Also, we define for each word $u$ a matrix $M(u)=\left(\begin{array}{ll}3^{|u|}&0\\\textrm{int}(u)&1\end{array}\right)$.
Then, since $\textrm{int}(uv)=3^{|v|}\textrm{int}(u)+\textrm{int}(v)$, we get that $M(u)M(v)=M(uv)$ and since the mapping $u\mapsto M(u)$
is also injective, it is an embedding of the semigroup $(\{0,1\}^+,\cdot)$ into $\bN^{2\times 2}$.

We define the following matrices $A_i$, $i\in[k]$, $B$ and $C$, all in $\bN^{6\times 6}$:
\begin{align*}
A_i &=\left(\begin{array}{lll}M(u_i)&0&0\\0&M(u_i)&0\\0&0&M(v_i)\end{array}\right),\\
B &=
\left(\begin{array}{llllll}
\mathrm{int}(u_1)&1&\mathrm{int}(u_1)&1&0&0\\
\mathrm{int}(u_1)&1&\mathrm{int}(u_1)&1&0&0\\
0&0&\mathrm{int}(u_1)&1&0&0\\
0&0&\mathrm{int}(u_1)&1&0&0\\
0&0&0&0&\mathrm{int}(v_1)&1\\
0&0&0&0&\mathrm{int}(v_1)&1\\
\end{array}\right),\\
C &= \left(\begin{array}{llllll}0&0&0&0&0&0\\0&0&0&0&0&0\\1&0&0&0&0&0\\0&0&0&0&0&0\\1&0&0&0&0&0\\0&0&0&0&0&0\end{array}\right),
\end{align*}
that is, $C$ has exactly two nonzero entries, namely $C_{3,1}=C_{5,1}=1$.

Then for any sequence $i_2,\ldots,i_t$, $t\geq 1$ we have
\[A_{i_2}\ldots A_{i_t}=\left(\begin{array}{lll}M(u)&0&0\\0&M(u)&0\\0&0&M(v)\end{array}\right)\]
with $u=u_{i_2}\ldots u_{i_t}$ and $v=v_{i_2}\ldots v_{i_t}$ and also
\[BA_{i_2}\ldots A_{i_t}=
\left(\begin{array}{llllll}
\mathrm{int}(u_1u)&\mathrm{int}(u)&\mathrm{int}(u_1u)&\mathrm{int}(u)&0&0\\
\mathrm{int}(u_1u)&\mathrm{int}(u)&\mathrm{int}(u_1u)&\mathrm{int}(u)&0&0\\
0&0&\mathrm{int}(u_1u)&\mathrm{int}(u)&0&0\\
0&0&\mathrm{int}(u_1u)&\mathrm{int}(u)&0&0\\
0&0&0&0&\mathrm{int}(v_1v)&\mathrm{int}(v)\\
0&0&0&0&\mathrm{int}(v_1v)&\mathrm{int}(v)\\
\end{array}\right),\]
and thus
\[BA_{i_2}\ldots A_{i_t}C=
\left(\begin{array}{llllll}
\mathrm{int}(u_1u)&0&0&0&0&0\\
\mathrm{int}(u_1u)&0&0&0&0&0\\
\mathrm{int}(u_1u)&0&0&0&0&0\\
\mathrm{int}(u_1u)&0&0&0&0&0\\
\mathrm{int}(v_1v)&0&0&0&0&0\\
\mathrm{int}(v_1v)&0&0&0&0&0
\end{array}\right),\]
which is synchronizing if and only if $u_1u_{i_2}\ldots u_{i_t}=v_1v_{i_2}\ldots v_{i_t}$, hence if $\{(u_i,v_i):i\in[k]\}$ is a positive
instance of FPCP, then $\mathcal{M}=\{A_i:i\in[k]\}\cup\{B,C\}$ is synchronizable.

For the other direction, suppose $\mathcal{M}$ is synchronizable. We already argued that any member $A$ of $\langle \{A_i:i\in[k]\} \rangle$ has the form $\left(\begin{array}{lll}M(u)&0&0\\0&M(u)&0\\0&0&M(v)\end{array}\right)$ for words $u,v$ with $u=u_{i_1}u_{i_2}\ldots u_{i_t}$ and
$v=v_{i_1}v_{i_2}\ldots v_{i_t}$ for some $i_j\in[k]$, $t\geq 0$. These matrices are clearly not (location) synchronizing.

Considering the matrix $C$, we have the following claims:

{\sl Claim A.} For any matrix $X$ we have $XC=\left(\begin{array}{l}c_1\\c_2\\\vdots\\c_6\end{array}{\textrm{\Huge{\quad$0$\quad}}}\right)$ for some $c_1,\ldots,c_6\in \bN$.

{\sl Claim B.} If $XCY$ is synchronizing for some matrices $X$ and $Y$, then so is $XC$.

Indeed, $XC$ is the matrix whose first column is the sum of the third and the fifth column of $X$, and whose other entries are all zero.
Also, if $XC=\left(\begin{array}{l}c_1\\c_2\\\vdots\\c_6\end{array}{\textrm{\Huge{\quad$0$\quad}}}\right)$ then
$XCY=\left(\begin{array}{l}c_1r_1\\c_2r_1\\\vdots\\c_6r_1\end{array}\right)$ where $r_1$ is the first row of $Y$.
If $XCY$ is synchronizing, this implies $c_ir_1=c_jr_1\neq 0$ for each $i,j\in[6]$, hence $c_i=c_j$ and $XC$ is synchronizing as well.

Thus, by ii) above we get that if $\mathcal{M}$ is synchronizable, then there is a synchronizing matrix of the form $XC$ with $X\in\langle \{A_i:i\in[k]\}\cup\{B\}\rangle$.

Inspecting members of $\langle\{A_i:i\in[k]\}\cup\{B\}\rangle$ we get the following claim:

{\sl Claim C.} Let $\mathcal{A}$ stand for the matrix semigroup $\langle\{A_i:i\in[k]\}\rangle$. Then for any $n\geq 0$, any member of
$\mathcal{A}(B\mathcal{A})^n$ has the form $\left(\begin{array}{lll}X&nX&0\\0&X&0\\0&0&Y\end{array}\right)$ for some
matrices $X,Y\in\bN^{2\times 2}$.

Indeed, for the base case $n=0$ we have matrices of the form $\left(\begin{array}{lll}M(u)&0&0\\0&M(u)&0\\0&0&M(v)\end{array}\right)$
satisfying the condition. Suppose the claim holds for $n$ and consider a matrix $M\in\mathcal{A}(B\mathcal{A})^{n+1}=\mathcal{A}(B\mathcal{A})^nB\mathcal{A}$. By the induction hypothesis, $M=M_0BA$ with $M_0=\left(\begin{array}{lll}X&nX&0\\0&X&0\\0&0&Y\end{array}\right)$,
and $A=\left(\begin{array}{lll}M(u)&0&0\\0&M(u)&0\\0&0&M(v)\end{array}\right)$ for some $X,Y\in\bN^{2\times 2}$ and words $u,v$.
We can also write $U_1$ for $\left(\begin{array}{ll}\mathrm{int}(u_1)&1\\\mathrm{int}(u_1)&1\end{array}\right)$
and $V_1$ for $\left(\begin{array}{ll}\mathrm{int}(v_1)&1\\\mathrm{int}(v_1)&1\end{array}\right)$.
Calculating the product we get
\begin{align*}
M=M_0BA
&=
\left(\begin{array}{lll}X&nX&0\\0&X&0\\0&0&Y\end{array}\right)
\left(\begin{array}{lll}
U_1&U_1&0\\ 0&U_1&0\\ 0&0&V_1
\end{array}\right)
\left(\begin{array}{lll}M(u)&0&0\\0&M(u)&0\\0&0&M(v)\end{array}\right)\\
&=
\left(\begin{array}{lll}XU_1M(u)&(n+1)XU_1M(u)&0\\ 0&XU_1M(u)&0\\ 0&0&YV_1M(v)\end{array}\right),
\\
\end{align*}
showing the claim.

Thus, since $\langle\{A_i:i\in[k]\}\cup\{B\}\rangle=\mathop\bigcup\limits_{n\geq 0}\mathcal{A}(B\mathcal{A})^n$, we get by Claim B that
if $\mathcal{M}$ is synchronizable, then there is a synchronizing matrix of the form $\left(\begin{array}{lll}X&nX&0\\0&X&0\\0&0&Y\end{array}\right)C$.
Writing $X=\left(\begin{array}{ll}x_1&x_2\\x_3&x_4\end{array}\right)$ and
$Y=\left(\begin{array}{ll}y_1&y_2\\y_3&y_4\end{array}\right)$ we get that this product is further equal to
$\left(\begin{array}{l}nx_1\\nx_3\\x_1\\x_3\\y_1\\y_3\end{array}\textrm{\quad\Huge{$0$}\quad}\right)$ which is synchronizing if and only if
$n=1$ and $x_1=x_3=y_1=y_3\neq 0$. By $n=1$ we get that if $\mathcal{M}$ is synchronizable, then there is a synchronizing matrix of the form
\[X=A_{j_1}A_{j_2}\ldots A_{j_\ell}BA_{i_2}A_{i_3}\ldots A_{i_t}C,\]
with $\ell\geq 0$, $t\geq 1$, $j_r,i_r\in[k]$.
Writing $u=u_1u_{i_2}\ldots u_{i_t}$, $v=v_1v_{i_2}\ldots v_{i_t}$, $u'=u_{j_1}\ldots u_{j_\ell}$ and $v'=v_{j_1}\ldots v_{j_\ell}$
we can write
\begin{align*}
X&=A_{j_1}A_{j_2}\ldots A_{j_\ell}BA_{i_2}A_{i_3}\ldots A_{i_t}C\\
&=
  \left(\begin{array}{lll}M(u')&0&0\\0&M(u')&0\\0&0&M(v')\end{array}\right)
  \left(\begin{array}{llllll}
  \mathrm{int}(u_1u)&1&\mathrm{int}(u_1u)&1&0&0\\
  \mathrm{int}(u_1u)&1&\mathrm{int}(u_1u)&1&0&0\\
  0&0&\mathrm{int}(u_1u)&1&0&0\\
  0&0&\mathrm{int}(u_1u)&1&0&0\\
  0&0&0&0&\mathrm{int}(v_1v)&1\\
  0&0&0&0&\mathrm{int}(v_1v)&1\\
  \end{array}\right)
  C\\
&=
\footnotesize{  \left(\begin{array}{llllll}
  3^{|u'|}\mathrm{int}(u_1u)&3^{|u'|}&3^{|u'|}\mathrm{int}(u_1u)&3^{|u'|}&0&0\\
  (\mathrm{int}(u')+1)\cdot\mathrm{int}(u_1u)&\mathrm{int}(u')+1&(\mathrm{int}(u')+1)\cdot\mathrm{int}(u_1u)&\mathrm{int}(u')+1&0&0\\
  0&0&3^{|u'|}\mathrm{int}(u_1u)&3^{|u'|}&0&0\\
  0&0&(\mathrm{int}(u')+1)\cdot\mathrm{int}(u_1u)&\mathrm{int}(u')+1&0&0\\
  0&0&0&0&3^{|v'|}\mathrm{int}(v_1v)&3^{|v'|}\\
  0&0&0&0&(|v'|+1)\cdot\mathrm{int}(v_1v)&\mathrm{int}(v')\\
  \end{array}\right)
  C}\\
&=\left(\begin{array}{l}
  3^{|u'|}\mathrm{int}(u_1u)\\(\mathrm{int}(u')+1)\cdot\mathrm{int}(u_1u)\\
  3^{|u'|}\mathrm{int}(u_1u)\\(\mathrm{int}(u')+1)\cdot\mathrm{int}(u_1u)\\
  3^{|v'|}\mathrm{int}(v_1v)\\(\mathrm{int}(v')+1)\cdot\mathrm{int}(v_1v)\\
  \end{array}\textrm{\Huge{\quad$0$\quad}}\right)
\end{align*}
which is synchronizing only if $3^{|u'|}=\mathrm{int}(u')+1$ and $3^{|v'|}=\mathrm{int}(v')+1$, that is, $u'=v'=\varepsilon$
implying $\ell=0$.

Hence if $\mathcal{M}$ is synchronizable then there exists a synchronizing product of the form
$BA_{i_2}A_{i_3}\ldots A_{i_t}C$, which in turn implies $u_1u_{i_2}\ldots u_{i_t}=v_1v_{i_2}\ldots v_{i_t}$, thus in that case
$\{(u_i,v_i):i\in[k]\}$ is indeed a positive instance of the FPCP problem.
\end{proof}

We note that the idea of encoding of a PCP variant within matrix semirings is not new, see e.g.~\cite{halavahirvensalo,potapov,gaubert}.
For example, $\mathbf{Z}$-mortality can be shown to be undecidable for $3\times 3$ integral matrices
via a similar embedding $(u,v)\mapsto M(u,v)=\left(\begin{array}{lll}4^{|u|}&0&0\\0&4^{|v|}&0\\\mathrm{int}(u)&\mathrm{int}(v)&1\end{array}\right)$ as in 
the proof of Theorem~\ref{thm-n-sync}, with $\mathrm{int}(u)$ being the base-4 value of $u$. This mapping is also an injective monoid
homomorphism. Then, defining $B=\left(\begin{array}{lll}0&0&0\\-1&0&-1\\0&0&0\end{array}\right)$ which satisfies
$B^2=B$ and $BM(u,v)B=(4^{|u|}+\mathrm{int}(u)-\mathrm{int}(v))B$ we get a similar construction (cf.~\cite{Halava97decidableand}),
also suitable for showing the undecidability of the zero-in-the-upper-left-corner problem.
However, the lack of substraction (in general, zero-sum-freeness of $\bN$) prevents us to apply this method.
Also, defining matrices of the form $TM(u,v)T^{-1}$ for a suitable $T$ (as in~\cite{Halava01mortalityin}, see also~\cite{paterson}) is again out of question
since in $\bN^{k\times k}$, only permutation matrices are invertible.
The most closest approach is that of the equal entries problem: in the proof we also showed undecidability of the
problem whether $\mathcal{A}$ generates a matrix having equal entries in the top-left corner and in entry $(5,5)$.
Actually, the embedding $(u,v)\mapsto\left(\begin{array}{ll}M(u)&0\\0&M(v)\end{array}\right)$ shows the same for $4\times 4$ matrices.
However, we were unable to modify the construction for $4\times 4$ matrices to \emph{shift} the values $\mathrm{int}(u)$ and
$\mathrm{int}(v)$ into, say, the first column and at the same time, \emph{overwrite} the values $3^{|u|}$ and $3^{|v|}$
by $\mathrm{int}(u)$ and $\mathrm{int}(v)$, respectively. (Adding them or something similar did not seem to work, either.)
That's why we had to use $6\times 6$ matrices -- it is quite plausible that the encoding is not the most compact possible
and the dimension can be further lowered.

\section{Conclusion, future directions}
We generalized the notion of synchronizability to automata with transitions weighted in an arbitrary semiring in two ways:
one of them, location synchronizability requires the existence of a word $u$ and a state $q$ such that starting from any state $p$,
$q$ and only $q$ has a nonzero weight after $u$ is being read; synchronizability additionally requires that this nonzero weight
is the same for all states $p$. In this paper we studied the \emph{complexity} of checking these properties, parametrized by the underlying
semiring.

Our results can be summarised as follows:
\begin{itemize}
\item Both problems are $\mathbf{PSPACE}$-hard for any nontrivial semiring.
\item For finite semirings, they are $\mathbf{PSPACE}$-complete.
\item For positive semirings, location synchronizability is $\mathbf{PSPACE}$-complete.
\item For locally finite semirings they are decidable (provided that the addition and product operations of the semiring are computable).
\item The mortality problem reduces to both problems in any semiring. Thus for semirings having an undecidable mortality problem,
  both variants of synchronization are undecidable. (This is the case for $\bZ$.)
\item If $(\{0,1\}^+,\cdot,\varepsilon)$ embeds into the multiplicative structure of $\bK$, then synchronizability is undecidable
  for $\bK$, even for deterministic automata.
\item Synchronizability is undecidable for any semiring where $1$ has infinite order in the additive semigroup. (This is the case for $\bN$.
  Note that for $\bN$, location synchronizability is in $\mathbf{PSPACE}$.)
\end{itemize}
We do not have any decidability results for $\bK$-synchronizability when the semiring $\bK$ is not locally finite, the element $1$ has a finite
order in the additive structure, and $\{0,1\}^+$ does not embed into the multiplicative semigroup. Also, it is not clear whether synchronizability
can be reduced to location synchronizability in general -- since in $\bN$, location synchronizability is decidable but synchronizability is
undecidable, so in general, synchronizability cannot be Turing-reduced to location synchronizability. It is also an interesting question whether
$\bN$-synchronizability of $5$-state automata is decidable or not -- we conjecture that it is still undecidable and one can use a slightly more compact
encoding of FPCP. Also, to cover the existing generalizations of synchronizability for the case of the probabilistic semiring,
we could study semirings that are equipped with a metric -- our current investigations can be seen as the case of this perspective
where the metric is the dicrete unit-distance metric.

\subsubsection*{Acknowledgement}
The research was supported by the European Union and the State of Hungary,
co-financed by the European Social Fund in the framework of T\'AMOP-4.2.4.A/ 2-11/1-2012-0001 ``National Excellence Program''.

The author is thankful to the anonymous referees for their suggestions and detailed comments.

\bibliographystyle{eptcs}
\bibliography{biblio-afl}
\end{document}